\documentclass{article}
\usepackage{arxiv}
\usepackage[OT2, T1]{fontenc}
\usepackage[russian, english]{babel}
\usepackage{caption}
\usepackage{lscape} 
\usepackage{amsmath}
\usepackage{geometry}
\usepackage{amsthm}
\usepackage{adjustbox}
\usepackage{amsfonts}
\usepackage{mathtools}
\usepackage{verbatim}
\usepackage{algpseudocode,algorithm,algorithmicx}
\usepackage{mathtools}
\usepackage{comment}
\usepackage{verbatim, times, booktabs, dcolumn, setspace, fullpage}

\newtheorem{theorem}{Theorem}

\newtheorem{definition}{Definition}

\usepackage{graphicx}
\usepackage{wrapfig}
\usepackage{caption}
\usepackage{mathtools}
\usepackage{float}
\usepackage{enumitem}

\usepackage[english,noautotitles-r]{SASnRdisplay} 
\usepackage{url}
\usepackage{subfigure}
\lstdefinestyle{r-output}{
style = r-style,
style = r-output-user,
}
\DeclareMathOperator\etr{etr}
\usepackage{url}            
\usepackage{booktabs}       
\usepackage{amsfonts}       
\usepackage{nicefrac}       
\usepackage{microtype}      
\usepackage{graphicx}
\usepackage{doi}

\title{Change point analysis - the empirical Hankel transform approach

}

\author{ \href{https://orcid.org/0000-0002-1964-7539}{\includegraphics[scale=0.06]{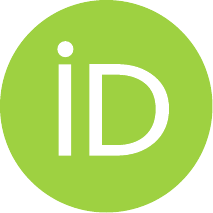}\hspace{1mm} Žikica Lukić} \\
	PhD student at the Faculty of Mathematics\\
	University of Belgrade\\
	Belgrade, 11000, Serbia \\
	\texttt{zikicamaster@gmail.com} \\
	\And
	\href{https://orcid.org/0000-0001-8243-9794}{\includegraphics[scale=0.06]{orcid.pdf}\hspace{1mm}Bojana Milošević} \\
    Faculty of Mathematics\\
	University of Belgrade\\
	Belgrade, 11000, Serbia \\
	\texttt{bojana.milosevic@matf.bg.ac.rs} \\
}

\begin{document}
\maketitle

\begin{abstract}
In this study, we introduce the first-of-its-kind class of tests for detecting change points in the distribution of a sequence of independent matrix-valued random variables. The tests are constructed using the weighted square integral difference of the empirical orthogonal Hankel transforms. The test statistics have a convenient closed-form expression, making them easy to implement in practice. We present their limiting properties and demonstrate their quality through an extensive simulation study. We utilize these tests for change point detection in cryptocurrency markets to showcase their practical use. The detection of change points in this context can have various applications in constructing and analyzing novel trading systems.
\end{abstract}

\section{Introduction}\label{sec::headings}
Change point analysis has garnered scientific interest in recent years due to numerous applications in finance \cite{andreou2002detecting, andreou2009structural, thies2018bayesian}, genetics \cite{chen2012parametric}, medicine \cite{contal1999application}, ecology \cite{achcar2010non}, climatology \cite{li2012multiple}, and other fields. Many of the methods outlined in the statistical literature focus on so-called parametric change point analysis, where a certain model is assumed and the existence of the change point corresponds to a change in the parameters of that model. This problem has been addressed in numerous works. We mention only a few examples, such as \cite{chakar2017robust, davis2006structural,fryzlewicz2014wild,tsukuda2014l2}.

Moreover, there are non-parametric change point methods that do not rely on any underlying assumptions regarding the model. In these cases, the primary focus is typically on detecting change points in the distribution of sample elements. Numerous papers address this issue, for example, \cite{brodsky1993nonparametric, dehling2013non, hawkins2010nonparametric}. One popular approach for nonparametric change point inference involves the use of two-sample test statistics. Specifically, a test statistic is formed using the following nonparametric analogue of CUSUM-type statistic, that consists of the following steps:
\begin{enumerate}[label=\arabic*)]
     \item The sample of size $n$ is divided into two subsamples of lengths $k$ and $n-k$, respectively;
    \item The two-sample statistic is calculated for selected subsamples, and multiplied by a suitable constant; 
    \item The value of the change point test statistic is taken to be the maximum of such values over $k$ ranging from 1 to $n$.
\end{enumerate}
This {idea, specifically designed for detecting change points, was introduced for the first time in \cite{huvskova2006change}. The authors utilized a characteristic-function-based two-sample test to implement this approach.}

Many papers are concerned with empirical characteristic function as a preferred transform, as in \cite{hlavka2020change,huvskova2006rank, meintanis2008testing, tan2016nonparametric}. However, to the best of our knowledge,  there are no results which utilize other types of integral transforms common in goodness-of-fit problems, like Laplace transform \cite{cuparic2022new,henze2012goodness, milovsevic2016new} or Hankel transform \cite{baringhaus2015two,baringhaus2010empirical,  hadjicosta2020Gamma}. Here, for the first time, we study the problem of detection of change in the distribution of matrix-valued sample elements.  The novel statistic is based on the recently proposed Hankel transform-based two-sample test, which has demonstrated efficient applicability for trading on cryptocurrency markets.

The paper is organized as follows. In Section \ref{sec::teststat}, we formally introduce the novel class of test statistics and represent it in a computationally convenient way. The asymptotic properties of the test statistics are then studied in Section \ref{sec::asym}. Section \ref{sec::power} is dedicated to small sample analysis and contains the results of a wide empirical study. The applicability of the test for detecting changes in cryptocurrency markets is illustrated in Section \ref{sec::realdata}. The Appendix contains the warp speed permutation bootstrap algorithm used in Section \ref{sec::power}, as well as the binary segmentation algorithm used for the detection of multiple change points in Section \ref{sec::realdata}.

\section{The test statistic}\label{sec::teststat}
Let us introduce some general definitions. The cumulative distribution function of the absolutely continuous matrix distribution on the cone of positive-definite symmetric matrices can be defined as follows (see, for instance, \cite{gupta2018matrix}):

\begin{equation*} F_X(A) = P(X < A) = \int_{0 < X < A} f(X; \theta) dX, \end{equation*}

where $f(X; \theta)$ represents the density function of $X$, potentially dependent on parameters $\theta \in \Theta$, and $A < B$ corresponds to the order relation defined as $A < B$ if $B-A$ is positive-definite. $A>0$ denotes that the matrix $A$ is positive-definite.


The notion of orthogonally invariant matrix-variate random variables was introduced in \cite{nas1}. Two matrix-variate random variables are considered orthogonally invariant in distribution (OID) if there exists an orthogonal matrix $\Gamma$ such that the distributions of matrices $A$ and $\Gamma B \Gamma'$ are the same. $B'$ denotes the matrix transpose of the matrix $B$. We refer to random variables satisfying this criterion as orthogonally invariant random variables and write $A\overset{OID}{=}B$. 

We modify the test statistic given in \cite{nas1} to address the change point type of problems. Let $X_1, X_2, \dots, X_n$ be the sample of independent symmetric positive definite random matrices, where $X_j$ has a cumulative distribution function $F_j$. 

    We want to test the hypothesis 
    \begin{equation}\label{nultahip}
     \begin{aligned}
        &H_0: X_1\overset{OID}{=}X_2\overset{OID}{=}\dots\overset{OID}{=}X_n\\
        &\quad \quad\quad\quad\quad\text{against the alternative}\\ &H_1: X_1\overset{OID}{=}X_2\overset{OID}{=}\dots \neq X_k\overset{OID}{=}X_{k+1}\overset{OID}{=}\dots \overset{OID}{=} X_n,
    \end{aligned}
    \end{equation} where the corresponding distributions of $X_1$ and $X_n$ are unknown and the index $k$ is unknown as well. In order to do that, we make use of the orthogonally invariant Hankel transform, defined as follows (see \cite{hadjicosta2020integral}):

\begin{definition}
Let $X>0$ be a random matrix with probability density function $f(X)$. For $\Re(\nu)>\frac{1}{2}(m-2)$ we define the orthogonally invariant Hankel transform of order $\nu$ as the function
\begin{align*}
    \Tilde{\mathcal{H}}_{X, \nu}(T)=E_X\big(\Gamma_m\big(\nu+\frac{1}{2}(m+1)\big)A_\nu(T, X)\big),
\end{align*}
where $T>0$, $\Gamma_m$ denotes the multivariate Gamma function and $A_\nu(T, X)$ denotes the Bessel function of the first kind of order $\nu$ with two matrix arguments. 
\end{definition}


{For further information on Bessel functions of matrix arguments, we refer to \cite{hadjicosta2020integral, herz1955bessel, muirhead2009aspects}.}

Its consistent estimator, the empirical orthogonal Hankel transform of order $\nu$,  is defined as (see \cite{hadjicosta2020integral}):
\begin{equation*}
   \Tilde{\mathcal{H}}_{n_1, \nu} (T)=\Gamma_m(\nu+\frac{1}{2}(m+1))\frac{1}{n_1}\sum\limits_{j=1}^{n_1} A_\nu (T, X_j) =\frac{1}{n_1}\sum\limits_{j=1}^{n_1} J_\nu (T, X_j).
\end{equation*}
For testing hypotheses \eqref{nultahip}, we propose the following test statistic:
\begin{align*}
        \mathcal{I}_{n, \gamma,  \nu}=&\max\limits_{1\leq k <n} \Big(\frac{k(n-k)}{n^2}\Big)^\gamma 
\frac{k(n-k)}{n}\int_{T>0}\Big(\Tilde{\mathcal{H}}_{{k}, \nu} (T)-\Tilde{\mathcal{H}}^0_{{n-k}, \nu} (T)\Big)^2dW(T),
\end{align*}
where $dW(T)$ is a standard Wishart measure (see \cite{hadjicosta2020integral, nas1}), $\Tilde{\mathcal{H}}_{{k}, \nu}$ denotes the empirical orthogonal Hankel transform of the first $k$ elements of the sample, while $\Tilde{\mathcal{H}}^0_{{n-k}, \nu}$ denotes the orthogonal empirical Hankel transform of the last $n-k$ elements of the sample. The test statistic can be expressed as:
\begin{align*}
        \mathcal{I}_{n, \gamma,  \nu}=&\max\limits_{1\leq k <n} \Big(\frac{k(n-k)}{n^2}\Big)^\gamma \frac{k(n-k)}{n}\Big(\frac{1}{k^2}\sum\limits_{i=1}^{k}\sum\limits_{j=1}^k \etr(-X_i-X_j)J_\nu(-X_i, X_j)+\\
     &\frac{1}{(n-k)^2}\sum\limits_{i=k+1}^n\sum\limits_{j=k+1}^n \etr(-X_i-X_j)J_\nu(-X_i, X_j)-\frac{2}{k(n-k)}\sum\limits_{i=1}^k\sum\limits_{j=k+1}^n \etr(-X_i-X_j)J_\nu (-X_i, X_j)\Big)\Big),
\end{align*}
where $\etr(X)=\exp(\text{Trace}(X))$ and $J_\nu(T,X) = \Gamma\big(\nu + \tfrac12 (d+1)\big) A_\nu(T,X)$.
 {Note that under the null hypothesis ($H_0$), a small value of the test statistic is expected. Hence, large values of the test statistic are considered significant.} In the next section, we present the asymptotic results of the novel test.
\section{Asymptotic results}\label{sec::asym}
In order to obtain the limiting null distribution of test statistic we first 
note that the test statistic $\mathcal{I}$ can be represented as  
$\mathcal{I}_{n, \gamma,  \nu}=\max\limits_{1\leq k\leq n} c_{n, k}(\gamma)I_{k, n-k, \nu}$, where $I_{k, n-k, \nu}$ is a two-sample test statistic introduced in \cite{nas1}. For the sake of brevity, we drop the parameter $\nu$ in the subsequent text and simply denote the test statistic as $I_{k, n-k}$.

\begin{theorem}\label{asimptotikaM}
    Let $X_1, X_2, \dots X_n$ be independent matrix-variate random variables orthogonally invariant in distribution, where $X_1\in F$. Let $\gamma \in (0, 1]$ and let $(\lambda_j)_{j=1}^\infty$ be the descending sequence of eigenvalues defined in \eqref{eigenvalues}.
Then the asymptotic distribution of $\mathcal{I}_{n, \gamma, \nu}$ is the same as that of
    \begin{equation*}
        \sup\limits_{t\in (0, 1)}(t(1-t))^\gamma \Big|(EJ_\nu(-X_1, X_1)-E \etr(-X_1-X_2)J_\nu (-X_1, X_2))+\sum\limits_{j=1}^\infty \lambda_j\Big(\frac{B_j^2(t)}{t(1-t)}-1\Big)\Big|, 
    \end{equation*}
    where $\{B_{j, n}(t),\; t\in (0,1)\}, j=1, 2, \dots$ are independent Brownian bridges.
\end{theorem}

\begin{proof}
We begin the proof by introducing some preliminary notions.
If we denote with $q(x, y)=\etr(-x-y)J_\nu (-x, y)$ and with $\Tilde{q}(x, y)=q(x, y)-E(q(Y, X_s)|Y=x)-E(q(X_r, Y)|Y=y)+Eq(X_r, X_s),\; r \neq s$, the following set of equalities holds:
\begin{equation}\label{nejednakosti}
E\Tilde{q}(X_r, X_s|X_r)=E\Tilde{q}(X_r, X_s|X_s)=E\Tilde{q}(X_r, X_s)=0.
\end{equation}
It is worth noting that the function $q(x, y)$ is symmetric in its arguments. Moreover, if we denote with $F$ the distribution function of $X_1$, we can easily establish
\begin{equation*}
    E\Tilde{q}^2(X_1, X_2)=\int_{X_1>0, X_2>0} \Tilde{q}^2(X_1, X_2)dF(X_1)dF(X_2)<\infty. 
\end{equation*}
From this, we conclude that there exist orthogonal eigenfunctions ${f_s(t),\; s=1, 2, \dots}$ and corresponding eigenvalues ${\lambda_s,\; s=1, 2, \dots}$, such that the following spectral approximation holds (see \cite{borovskikh1994theory, serfling2009approximation}): 
\begin{align}
   \lim\limits_{K\to\infty} \int_{X_1>0, X_2>0} (\Tilde{q}^2(X_1, X_2)-\sum\limits_{s=1}^K\lambda_s f_s(X_1)f_s(X_2))^2 dF(X_1)dF(X_2)=0;\label{eigenvalues}
   \end{align}
   \begin{align*}
  & \int_{X_1>0} f^2_s(X_1)dF(X_1)=1, \; s=1, 2, \dots;\\
  & \int_{X_1>0} f_r(X_1)f_s(X_1)dF(X_1)=0, \;r\neq s=1, 2, \dots,
 \end{align*}
  from  which we obtain
  \begin{align*}
  &E\Tilde{q}^2(X_1, X_2)=\int_{X_1>0, X_2>0} \Tilde{q}^2(X_1, X_2)dF(X_1)dF(X_2)=\sum\limits_{j=1}^\infty \lambda_j^2.
\end{align*}
From now on, we break the proof into several steps. 

In Step (I),
we decompose the test statistic into four parts,  and show that two of them do not influence the limiting distribution. 
We  assert this by making repeated use of the H\'ajek-R\'enyi inequality to majorate the terms.

In Step (II), we determine the finite-dimensional distributions of the limiting process, by establishing the approximation formula, and making use of the Donsker theorem.

In Step (III), using the properties of the Wiener process, we determine the asymptotic distribution in question and finalize the proof. 
\medskip

\noindent\underline{Step (I):} The following decomposition holds:
\begin{equation*}
    I_{k, n-k}=C_{k1}+C_{k2}+C_{k3}+C_{k4},
\end{equation*}
where 
\begin{align*}
    &C_{k1}=\frac{n}{k(n-k)}\Big(\frac{1}{k}\sum\limits_{v=1}^k\sum\limits_{s=1, s\neq v}^k \Tilde{q}(X_v, X_s)+\frac{1}{n-k}\sum\limits_{v=k+1}^n\sum\limits_{s=k+1, s\neq v}^n \Tilde{q}(X_v, X_s)-\frac{1}{n}\sum\limits_{v=1}^n\sum\limits_{s=1, s\neq v}^n \Tilde{q}(X_v, X_s)\Big), \\
    &C_{k2}=\frac{n}{k(n-k)}(EJ_\nu(-X_i, X_i)-Eq(X_1, X_2)),\\
    &C_{k3}=-\frac{2}{k^2}\sum\limits_{r=1}^k(Eq((X_r, X_s)|X_r)-Eq(X_1, X_2))-\frac{2}{(n-k)^2}\sum\limits_{r=k+1}^n (Eq((X_r, X_s)|X_r)-Eq(X_1, X_2)),\\
    &C_{k4}=\frac{1}{k^2}\sum\limits_{i=1}^k (J_\nu(-X_i, X_i)-EJ_\nu(-X_i, X_i))+\frac{1}{(n-k)^2}\sum\limits_{i=k+1}^n (J_\nu(-X_i, X_i)-EJ_\nu(-X_i, X_i)). 
\end{align*}

 The main idea is to demonstrate that $C_{k3}$ and $C_{k4}$ have no impact on the limiting distribution, while the non-random term $C_{k2}$ and the random term $C_{k1}$ do. Then, we establish the limiting distribution of $C_{k1}$, which is not a straightforward task and depends on the ability to determine the distribution of an arbitrarily close approximation. The final part of Step (I) consists of demonstrating that the asymptotic distributions of $C_{k1}$ and the approximation are identical.

We begin by showing that
$C_{k3}$ and $C_{k4}$ do not influence the limiting distribution. 

We can see that if $X$ is a random variable independent of $X_r$ having the same distribution function $F$, random variables
\begin{equation*}
    L_r=E(q(X_r, X)|X_r)-Eq(X_1, X_2), r=1, 2, \dots, n
\end{equation*}
are independent and identically distributed with zero mean and finite variance. 
Note that while $X_i$ are matrix random variables, random variables $L_r$ are real-valued. By applying the H\'ajek-R\'enyi inequality, we obtain that for every positive real constant $A$ and every $\gamma \in [0, 1]$, the following inequality holds:
\begin{equation*}
    P_{H_0}\Big(\max\limits_{1\leq k < n} c_{k, n}(\gamma) \frac{2}{k^2}\Big|\sum\limits_{r=1}^k L_r\Big | \geq A\Big)\leq \frac{4Var(L_1)}{A^2 n^{2\gamma}}\sum\limits_{k=1}^n \frac{1}{k ^{2(1-\gamma)}}.
\end{equation*}
The last term can be upper-bounded by $D_3n^{-\min({2\gamma}, 1)}(1+Ind(\gamma\geq\frac{1}{2})\log(n))$, where $D_3$ is a constant greater than zero. Similarly, we also have
\begin{equation*}
    P_{H_0}\Big(\max\limits_{1\leq k < n} c_{k, n}(\gamma) \frac{2}{(n-k)^2}\Big|\sum\limits_{r=k+1}^n L_r\Big | \geq A\Big)\leq D_3n^{-\min({2\gamma}, 1)}(1+Ind(\gamma\geq \frac{1}{2})\log(n)).
\end{equation*}
Therefore, given that the above inequalities hold:
\begin{equation*}
    \max\limits_{1\leq k<n} c_{k, n}(\gamma)|C_{k3}|=o_P(1), n\to\infty, 
\end{equation*}
and the term $|C_{k3}|$ does not impact the limiting distribution.

Next, we demonstrate that $C_{k4}$ does not impact the limiting behavior. We note that the random variables

\begin{equation*}
    K_r=J_\nu(-X_r,X_r)-EJ_\nu (-X_r, X_r), r=1, 2, \dots, n
\end{equation*}
are independent and identically distributed with zero mean and finite variance. By applying the H\'ajek-R\'enyi inequality, we obtain that for every real constants $A>0$ and $\gamma \in [0, 1]$, the following inequality holds:
\begin{equation*}
    P_{H_0}\Big(\max\limits_{1\leq k < n} c_{k, n}(\gamma) \frac{1}{k^2}\Big|\sum\limits_{r=1}^k K_r\Big | \geq A\Big)\leq \frac{Var(K_1)}{A^2 n^{2\gamma}}\sum\limits_{k=1}^n \frac{1}{k ^{2(1-\gamma)}}.
\end{equation*}
The last term can be upper-bounded by $D_4 n^{-\min({2\gamma}, 1)}(1+Ind(\gamma\geq \frac{1}{2})\log(n))$, where $D_4$ is a positive constant. 
\begin{equation*}
    P_{H_0}\Big(\max\limits_{1\leq k < n} c_{k, n}(\gamma) \frac{1}{(n-k)^2}\Big|\sum\limits_{r=k+1}^n K_r\Big | \geq A\Big)\leq D_4 n^{-\min({2\gamma}, 1)}(1+Ind(\gamma\geq \frac{1}{2})\log(n)).
\end{equation*}
Hence, it follows that
\begin{equation*}
    \max\limits_{1\leq k<n} c_{k, n}(\gamma)|C_{k4}|=o_P(1),\; n\to\infty, 
\end{equation*}
and the term $|C_{k4}|$ does not affect the limiting distribution.
Furthermore, it is worth noting that
\begin{align*}
Var(C_{k1})=O\Big(2\Big(\frac{n}{k(n-k)}\Big)^2E\Tilde{q}^2(X_1, X_2)\Big).
\end{align*}
Since  $\sqrt{Var(C_{k1})}$ and $C_{k2}$ are of the same order, it remains to determine the limiting distribution of $C_{k1}$ in Step (II). 

\underline{Step (II):} 
{Consider the auxiliary statistic:
\begin{equation*}
    S_k(\Tilde{q})=\sum\limits_{1\leq i < j \leq k} \Tilde{q}(X_i, X_j), \; k = 1, 2, \dots, n.
\end{equation*}
Under $H_0$, we have:
\begin{equation}\label{uslovMart}
    E(S_{k+1}(\Tilde{q}) | X_1, X_2, \dots, X_k)=S_k(\Tilde{q})+\sum\limits_{j=1}^k E(\Tilde{q}(X_i, X_{k+1})|X_1, X_2, \dots, X_k)=S_k(\Tilde{q}), \; k=1, 2, \dots, n-1.
\end{equation}
We establish that $\{(S_k(\Tilde{q}), \sigma(X_1, X_2, \dots, X_k)), k=1, 2, \dots, n\}$ is a martingale. $\sigma(X_1, X_2, \dots, X_k)$ denotes the $\sigma$-algebra generated by $X_1, X_2, \dots, X_k$.
}
Applying the H\'ayek-R\'enyi inequality yields the following result:
{
\begin{align*}
    &P(\max\limits_{1\leq k < n} n\Big(\frac{k(n-k)}{n^2}\Big)^{\gamma+1} \frac{1}{k^2}|S_k(\Tilde{q})|\geq A)\leq P(\max\limits_{1\leq k < n} n\Big(\frac{k}{n}\Big)^{\gamma+1} \frac{1}{k^2}|S_k(\Tilde{q})|\geq A)\leq \\&\frac{1}{n^{2\gamma+2}}\frac{1}{A^2}\sum\limits_{k=1}^{n-1} \frac{1}{k^{2\gamma-2}}E(S_k(\Tilde{q})-S_{k-1}(\Tilde{q}))^2\leq
    \frac{1}{n^{2\gamma}A^2}\sum\limits_{k=1}^{n-1}\frac{1}{k^{2-2\gamma}} E(S_k(\Tilde{q})-S_{k-1}(\Tilde{q}))^2\\
    &\leq \frac{D_0}{A^2}E\Tilde{q}^2(X_1, X_2),
\end{align*}
for some positive constant $D_0$. Repeating the computation for the function $\Tilde{S}_k(\Tilde{q})=\sum\limits_{k+1\leq i<j\leq n} \Tilde{q}(X_i, X_j)$, which satisfies \eqref{uslovMart}, and combining the results, we obtain:
}
\begin{equation*}
    P(\max\limits_{1 \leq k < n} c_{k, n}(\gamma)|C_{k1}|\geq A)\leq \frac{D_1}{A^2} E\Tilde{q}^2(X_1, X_2), 
\end{equation*}
for some positive constant $D_1$. This equality holds true if we replace $\Tilde{q}$ with any function that satisfies \eqref{nejednakosti}. We utilize the function $\Tilde{q}-\Tilde{q}_K$, such that $\Tilde{q}_K$ is defined by:
\begin{equation*}
    \Tilde{q}_K(x, y)=\sum\limits_{s=1}^K \lambda_s f_s(x)f_s(y),
\end{equation*}
where $(\lambda_s, f_s)$ are the (sorted) eigenvalues and eigenfunctions defined in \eqref{eigenvalues} and $K$ denotes the arbitrary natural number. 

Moreover, by utilizing the fact that
\begin{equation*}
    E\Tilde{q}^2(X_1, X_2) = \sum\limits_{k=1}^\infty \lambda^2_k,
\end{equation*}
and applying the H\'ayek-R\'enyi inequality once more, we establish that
\begin{equation}\label{nejednakostC}
    P\Big(\max\limits_{1\leq k < n} c_{k, n}(\gamma) |C_{k1}- C_{k1}(K)|\geq A\Big)\leq \frac{D_5}{A^2}E(\Tilde{q}(X_1, X_2)-\Tilde{q}_K(X_1, X_2))^2=
    \frac{D_5}{A^2}\sum\limits_{j=K+1}^\infty \lambda_j^2,
\end{equation}
where $D_5$ is a positive constant and $n\geq 2$ and $K\in \mathbb{N}$ and $C_{k1}(K)$ denotes $C_{k1}$ where $\Tilde{q}$ is replaced with $\Tilde{q}_K$. We can now conclude that
\begin{align*}
    \frac{k(n-k)}{n}C_{k1}(K)&=\sum\limits_{s=1}^K\lambda_s\Big(\frac{n}{k(n-k)}\Big(\sum\limits_{j=1}^k f_s(X_j)-\frac{k}{n}\sum\limits_{j=k+1}^n f_s(X_j)\Big)^2\\& - \frac{1}{n}\Big(\frac{n-k}{k}\sum\limits_{j=1}^k f^2_s(X_j)+\frac{k}{n-k}\sum\limits_{j=1}^n f^2_s(X_j)\Big)\Big). 
\end{align*}
Using the Donsker theorem \cite{suquet2009reproducing}, we obtain:
\begin{equation*}
    \Big\{\frac{1}{\sqrt{n}}\sum\limits_{i=1}^{[nt]} \mathbf{f_K}(X_i), t\in(0,1)\Big\}\xrightarrow[n\to\infty]{D((0, 1))} {\mathbf{W_K}(t), t\in(0,1)},
\end{equation*}
where $\mathbf{W_K}=(W_1, \dots, W_K)$ is a $K$-dimensional Wiener process with independent components in $\mathbb{R}^\mathbf{K}$ and $\mathbf{f_K}=(f_1, f_2, f_3, \dots, f_K)$ is a vector of eigenfunctions up to the position $K$. Therefore, we establish that the asymptotic distribution of $\max\limits_{1\leq k < n} \Big(\frac{k(n-k)}{n^2}\Big)^{\gamma}C_{k1}(K)$ is given by:
\begin{equation*}
\max\limits_{1\leq k < n} \Big(\frac{k(n-k)}{n^2}\Big)^{\gamma}\sum\limits_{s=1}^K\lambda_s\Big(\frac{n^2}{k(n-k) }(W_s(k/n)-\frac{k}{n}W_s(1))^2-1\Big).
\end{equation*}

\underline{Step (III):} Note that the random processes $\lambda_s\Big(\frac{n^2}{k(n-k) }(W_s(k/n)-\frac{k}{n}W_s(1))^2-1\Big)$ are independent for $s=K+1, K+2, \dots$.  Additionally, these random processes have a zero mean.
Furthermore, as their variances are finite, we can apply the H\'ajek-R\'enyi inequality, yielding the following:
\begin{align*}
     P\Big(\max\limits_{1\leq k < n} \Big(\frac{k(n-k)}{n^2}\Big)^{\gamma}\sum\limits_{s=K+1}^\infty \lambda_s\Big(\frac{n^2}{k(n-k) }(W_s(k/n)-\frac{k}{n}W_s(1))^2-1\Big) \geq A\Big)
   \leq \frac{D_5}{A^2}\sum\limits_{s=K+1}^\infty \lambda_s^2.
\end{align*}

From \eqref{nejednakostC}, we can observe that for a sufficiently large value of $K$, the asymptotic distributions of $C_{k1}$ and $C_{k1}(K)$ are close to each other. Consequently, by taking the limit as $K$ tends to infinity, we conclude that the asymptotic distribution of $\max\limits_{1\leq k < n}\Big(\frac{k(n-k)}{n^2}\Big)^{\gamma} C_{k1}$ is identical to that of
\begin{equation*}
   \max\limits_{t\in (0, 1)} \Big(t(1-t)\Big)^{\gamma}\sum\limits_{s=1}^\infty \lambda_s\Big(\frac{(W_s(t)-tW_s(1))^2}{t(1-t)}-1\Big).
\end{equation*}
This concludes the proof.
\end{proof}

Since the limiting null distribution of novel statistics is not free of distribution of $X_1$, for the derivation of p-values,  we suggest the usage of the permutation bootstrap algorithm from \cite{huvskova2006change}. This approach can be theoretically justified in a similar manner for our test.

\section{Power study}\label{sec::power}

In this section, we present the results of the power study. 
All computations are done using MATLAB \cite{MATLAB}. The algorithm developed in \cite{koev2006efficient} was implemented to evaluate the Bessel functions of two matrix arguments.

In our study, we fix the value of parameter $\nu$ to $\nu=1$, as is common practice in problems of this nature (see \cite{baringhaus2010empirical, nas1}). 

The empirical powers in various settings, for the level of significance $\alpha=0.05$, are obtained using a warp-speed modification of the permutation bootstrap algorithm  - Algorithm \ref{WSBA} (see, e.g., \cite{giacomini2013warp}), with $N=500$ replications.

The following distributions are considered:
\begin{enumerate}
    \item Wishart distribution with the shape parameter $a$ and the scale matrix $\Sigma$, denoted by $W_d(a, \Sigma)${, with a density
    \begin{equation*}
    f_{W, a, \Sigma}(X) = \frac{1}{\Gamma_d(a)}(\det\Sigma)^a(\det X)^{a-\frac{1}{2}(d+1)}\etr(-\Sigma X);
    \end{equation*}
    
    }
    \item Inverse Wishart distribution with the shape parameter $a$ and the scale matrix $\Sigma$, denoted by $IW_d(a, \Sigma)$, { with a density
    \begin{equation*}
        f_{IW, a, \Sigma} (X)=\dfrac{(\det \Sigma)^{\frac{a}{2}}\etr(-\frac{1}{2}\Sigma X^{-1})}{2^{\frac{a d}{2}} \Gamma_d(\frac{a}{2})(\det X)^{\frac{a+d+1}{2}}};
    \end{equation*}
    }
    \item Sample covariance matrix distributions obtained from the uniform vectors $(U_1, \dots, U_d)$, where $U_i\in \mathcal{U}[0, 1]$, denoted by $CMU_d$, { with a density
    \begin{equation*}
        f_{(U_1, \dots, U_d)} ((x_1, \dots, x_d))=1, \; x_i\in [0, 1],\; 1\leq i \leq d;
    \end{equation*}
}
 \item Sample covariance matrix distributions obtained from the random vectors having the multivariate  $t$ distribution  with $a$ degrees of freedom, denoted by $CMT_d(a, \Sigma)$, { with a density
\begin{equation*}
        f_{t, a} (x)=\frac{1}{(\det(\Sigma))^\frac{1}{2}}\frac{\Gamma(\frac{a+d}{2})}{\Gamma(\frac{d}{2})(a\pi)^\frac{d}{2}}(1+\frac{x'\Sigma^{-1}x}{a})^{-\frac{a+d}{2}}.
    \end{equation*} }
\end{enumerate}

In all cases above, we assume that $d$ represents the dimension of the respective matrices. When estimating sample covariance matrices, samples of dimension $d$ have been considered.

The empirical powers for $n=40$ and $n=100$, and different positions of the change point ($k$) and different values of the tuning parameter $\gamma$, for $d=2$ and $d=3$, are presented in Tables \ref{pow2} and \ref{pow3}, respectively.
From the results presented therein, we can see that the empirical sizes are close to the significance level, with just a few notable deviances that are decreasing with the sample size. Also notable is that the tests are not very sensitive to the value of the tuning parameter $\gamma$ - in most cases, discrepancies are not more than 10\%. This behavior is anticipated as the similar test statistics demonstrate robustness to variations in the parameter $\gamma$ \cite{huvskova2006change}.

On the other hand, the location of the change point seems to play an important role in its detection. In most cases, the central change points are easier to detect than those close to the boundaries. In asymmetric cases, optimizing the parameter $\gamma$ could enhance the test power. It is worth noting that the change point detection took place in cases where distributions shared the same expectation.

The loss of power with the increase of dimension is expected. However, the ratio between the rejection rate among the different alternatives remains approximately the same.

    \begin{table}[htbp] 
\caption{Powers of the test for different sample sizes and change point locations: $2\times 2$ matrices, $\gamma=0.5$ ($\gamma=1$).}\label{pow2}
\begin{adjustbox}{width=\linewidth,center}

\begin{tabular}{l|lllllllllll}
\hline
&\multicolumn{11}{c}{$n=40, k=20$}\\
Distributions & $W_{2}(2.5, I_2)$ & $IW_{2}(2.5, I_2)$ & $CMT_2(1, I_2)$ &$CMU_2$ & $W_{2}(2.5, 2I_2)$ & $IW_{2}(4, 2.5I_2)$ & $W_2(2.5, K_2)$ & $CMT_2(3, K_2)$ & $CMT_2(5, K_2)$ & $CMT_2(3, I_2)$ & $CMT_2(5, I_2)$ \\ \hline
$W_{2}(2.5, I_2)$ & 7 (4) & 29 (29) & 12 (11) & 100 (100) & 16 (23) & 27 (32) & 18 (33)  & 51 (52) & 57 (62) & 68 (71) & 80 (85) \\
$IW_{2}(2.5, I_2)$ &  & 4 (3) & 9 (12)& 100 (100) & 62 (68) & 6 (8) & 5 (6) & 19 (21) & 23 (26) & 32 (34) & 42 (49)\\
$CMT_2(1, I_2)$ &  &  & 6 (7) & 100 (100)  & 23 (35) & 4 (7) & 9 (7) & 23 (29) & 34 (38) & 47 (48) & 59 (60)\\
$CMU_2$ &  &  &  & 7 (4) & 100 (100) & 100 (100)  & 100 (100) & 100 (100) & 99 (99) & 98 (99) & 97 (98) \\
$W_{2}(2.5, 2I_2)$ &  &  &  &  & 6 (4) & 78 (82) & 67 (68) & 77 (80) & 85 (86) & 86 (91) & 93 (95) \\
$IW_{2}(4, 2.5I_2)$ &  &  &  &  &  &  4 (6) & 7 (10) & 26 (27) & 35 (43) & 41 (44) & 58 (57) \\ 
$W_2(2.5, K_2)$ &  &  &  &  &  &   & 8 (5)  & 21 (19) & 27 (27) & 32 (32) & 45 (45)  \\
$CMT_2(3, K_2)$ &  &  &  &  &  &   &   & 3 (3) & 4  (5)& 11 (6)  & 12 (11) \\
$CMT_2(5, K_2)$ &  &  &  &  &  &   &   &  & 2 (3) & 7 (7) & 11 (10) \\
$CMT_2(3, I_2)$ &  &  &  &  &  &   &   &  &  & 3 (4) & 7 (6) \\
$CMT_2(5, I_2)$ &  &  &  &  &  &   &   &  &  & & 7 (5) \\\hline
&\multicolumn{11}{c}{$n=40, k=10$}\\


Distributions& $W_{2}(2.5, I_2)$ & $IW_{2}(2.5, I_2)$ & $CMT_2(1, I_2)$ &$CMU_2$& $W_{2}(2.5, 2I_2)$ & $IW_{2}(4, 2.5I_2)$ & $W_2(2.5, K_2)$ &  $CMT_2(3, K_2)$ & $CMT_2(5, K_2)$ & $CMT_2(3, I_2)$ & $CMT_2(5, I_2)$\\ \hline
$W_{2}(2.5, I_2)$ & 4  (4)& 12 (13) & 7 (10) & 100 (100) & 18 (23) & 18 (16) & 10 (15) & 20 (21) & 26 (29) & 32 (35) &  49 (47) \\
$IW_{2}(2.5, I_2)$ & 25 (24) & 6 (4) & 6 (5) & 100 (100) & 55 (54) & 10 (5) & 6 (9) & 7 (6) & 9 (14) & 16 (20) &  26 (17) \\
$CMT_2(1, I_2)$ & 14 (10) & 5 (4) & 5 (6)  & 100 (100)  & 23 (24) & 8 (6) & 10 (6) & 10 (13)& 17 (17) & 32 (20) & 34 (31) \\
$CMU_2$ & 100 (100) & 100 (100) & 100 (100) & 5 (5) & 100 (100) & 100 (100) & 100 (100)  & 98 (96) & 98 (97) & 90 (85) &  82 (85)\\
$W_{2}(2.5, 2I_2)$ & 5 (9) & 28 (28) & 5 (6) & 100 (100)  & 3 (7) & 50 (43) & 34 (32) & 34 (37) & 53 (45) &  63 (55) &  70 (65)\\
$IW_{2}(4, 2.5I_2)$ & 19 (16) & 6 (3) & 3 (5) & 100 (100)  & 58 (64) & 4 (6) & 7 (4) & 10 (11) & 15 (14) & 13 (14) & 26 (24) \\
$W_2(2.5, K_2)$ & 19 (24) & 5 (7) & 3 (7) & 100 (100) & 66 (57) & 7 (7) & 5 (3) & 8 (11) & 20 (14) & 19 (18) & 27 (26) \\
$CMT_2(3, K_2)$ & 46 (43) & 16 (14) & 17 (14) & 98 (97) & 65 (65) & 26 (24)  & 21 (18)  & 6  (5)& 6 (8) & 6 (6) & 6 (8)\\
$CMT_2(5, K_2)$ & 59 (53) & 25 (23) & 23 (23) & 96 (96) & 75 (69) & 28 (28)  & 29 (20)  & 8 (6) & 4 (3) &  4 (5) &  8 (6) \\
$CMT_2(3, I_2)$ & 59 (53) & 28 (33) & 29 (35) & 92 (91) & 76 (78) & 36  (37) & 36 (23)  & 8 (8) & 12 (9) & 3 (3) & 2 (4) \\
$CMT_2(5, I_2)$ & 70 (65) & 36 (41) & 45 (43) & 88 (88) & 88 (81) & 52  (50) & 41 (33)  & 8 (7) & 7 (7) & 6 (6) & 5 (5) \\
\hline




\hline 
&\multicolumn{11}{c}{$n=100, k=50$}\\
Distributions & $W_{2}(2.5, I_2)$ & $IW_{2}(2.5, I_2)$ & $CMT_2(1, I_2)$ &$CMU_2$ & $W_{2}(2.5, 2I_2)$ & $IW_{2}(4, 2.5I_2)$ & $W_2(2.5, K_2)$ & $CMT_2(3, K_2)$ & $CMT_2(5, K_2)$ & $CMT_2(3, I_2)$ & $CMT_2(5, I_2)$ \\ \hline
$W_{2}(2.5, I_2)$ & 7 (5) & 64 (67) & 19 (22) & 100 (100) & 53 (65) & 65 (67) & 75 (74)  & 94 (94) & 99 (99) & 99 (99) & 100 (100) \\
$IW_{2}(2.5, I_2)$ &  & 6 (5) & 11 (16) & 100 (100) & 99 (99) & 9 (9) & 8 (10) & 33 (51) & 58 (58) & 70 (78) & 90 (92)\\
$CMT_2(1, I_2)$ &  &  & 5 (6) & 100 (100)  & 54 (61) & 14 (14) & 19 (25) & 59 (63) & 72 (77) & 90 (90) & 94 (96)\\
$CMU_2$ &  &  &  & 5 (6) & 100 (100) & 100 (100)   & 100 (100) & 100 (100) & 100 (100) & 100 (100) & 100 (100) \\
$W_{2}(2.5, 2I_2)$ &  &  &  &  & 4 (4) & 100 (100) & 99 (100) & 99 (100) & 100 (100) & 100 (100) & 100 (100) \\
$IW_{2}(4, 2.5I_2)$ &  &  &  &  &  &  7 (6) & 10 (9) & 51 (54) & 78 (79) & 85 (89) & 93 (96) \\ 
$W_2(2.5, K_2)$ &  &  &  &  &  &   & 4 (4)  & 42 (53) & 61 (65) & 73 (81)& 87 (86)  \\
$CMT_2(3, K_2)$ &  &  &  &  &  &   &   & 4 (6)  & 7 (6) & 18 (21)  & 31 (31) \\
$CMT_2(5, K_2)$ &  &  &  &  &  &   &   &  & 6 (5)& 9 (10) & 16 (17) \\
$CMT_2(3, I_2)$ &  &  &  &  &  &   &   &  &  & 4 (4) & 10 (10) \\
$CMT_2(5, I_2)$ &  &  &  &  &  &   &   &  &  & & 6 (6)\\
\hline 




&\multicolumn{11}{c}{$n=100, k=25$}\\
Distributions& $W_{2}(2.5, I_2)$ & $IW_{2}(2.5, I_2)$ & $CMT_2(1, I_2)$ &$CMU_2$ & $W_{2}(2.5, 2I_2)$ & $IW_{2}(4, 2.5I_2)$ & $W_2(2.5, K_2)$ & $CMT_2(3, K_2)$ & $CMT_2(5, K_2)$ & $CMT_2(3, I_2)$ & $CMT_2(5, I_2)$ \\ \hline
$W_{2}(2.5, I_2)$ & 7 (4) & 58 (51) & 23 (22) & 100 (100) & 17 (24) & 41 (49)& 61 (66)  & 86 (85) & 93 (92) & 95 (95) & 100 (99) \\
$IW_{2}(2.5, I_2)$ & 38 (45) & 7 (7) & 12 (11) & 100 (100) & 85 (87) & 6 (6) & 6 (7) & 27 (25) & 47 (45) & 59 (59) & 76 (73)\\
$CMT_2(1, I_2)$ & 8 (7) & 11 (10)& 4 (3) & 100 (100)  & 24 (32)& 5 (9) & 9 (13) & 43 (46) & 58 (58) & 71 (69) & 84 (87)\\
$CMU_2$ & 100 (100) & 100 (100) & 100 (100)  & 3 (3) & 100 (100) & 100 (100)  & 100 (100) & 100 (100) & 100 (100) & 100 (100) & 100 (100) \\
$W_{2}(2.5, 2I_2)$ & 46 (46)  & 94 (94) & 56 (59) & 100 (100) & 3 (5) & 99 (98) & 97 (98) & 99 (99) & 99 (99) & 100 (100) & 100 (100) \\
$IW_{2}(4, 2.5I_2)$ & 39 (33) & 8 (9)& 10 (10) & 100 (100) & 97 (97) &  4 (3)& 11 (8) & 49 (48)& 65 (65)& 79 (76) & 91 (90) \\ 
$W_2(2.5, K_2)$ & 29 (39) & 5 (5) & 13 (12) & 100 (100) & 93 (94) & 8 (9) & 7 (5)  & 30 (29)& 50 (51) & 68 (62)& 72 (73)  \\
$CMT_2(3, K_2)$ & 73 (58) & 17 (18) & 39 (43) & 100 (100) & 98 (95) & 28 (29)  & 21 (15)  & 5 (4) & 8 (7) & 12 (16)  & 23 (16) \\
$CMT_2(5, K_2)$ & 87 (88) & 30 (33) & 55 (50) & 100 (100) & 100 (99) & 45 (36) &  32 (34) & 5 (4) & 8 (6) & 9 (6) & 8 (14) \\
$CMT_2(3, I_2)$ & 89 (90) & 40 (43) & 68 (64) & 100 (100) & 99 (100) & 51 (60)  & 46 (43)  & 12 (10) & 7 (6) & 5 (7) & 6 (5)\\
$CMT_2(5, I_2)$ & 98 (98) & 61 (58) & 82 (75) & 100 (100) & 100 (100) & 73 (76)  & 65 (69)  & 15 (18) & 10 (10) & 6 (6) & 7 (7)\\
\hline


\hline
\end{tabular}
\end{adjustbox}
\end{table}

\begin{table}[htbp]
\caption{Powers of the test for different sample sizes and change point locations: $3\times 3$ matrices, $\gamma=0.5$ ($\gamma=1$).}\label{pow3}
\begin{adjustbox}{width=\linewidth,center}
\begin{tabular}{l|lllllllllllll}
\hline
&\multicolumn{11}{c}{$n=40, k=20$}\\
Distributions & $W_3(3, I_3)$ & $IW_3(3, I_3)$ & $CMT_3(1, I_3)$ &$CMU_3$& $W_3(3, 2I_3)$ & $IW_3(5, 3I_3)$ & $W_3(3, K_3)$ & $CMT_3(3, K_3)$ & $CMT_3(5, K_3)$  & $CMT_3(3, I_3)$ & $CMT_3(5, I_3)$\\ \hline
$W_3(3, I_3)$ & 5 (4) & 16 (19) & 7 (6) & 100 (100) & 9 (8) & 22 (42) & 8 (10) & 26 (38) & 36 (49) & 27 (38) & 40 (45) \\
$IW_3(3, I_3)$ &  & 5 (5)  & 7 (12) & 100 (100) & 15 (29)  & 6 (5) & 21 (27) & 11 (10) & 16 (17) & 14 (21) & 23 (24)\\
$CMT_3(1, I_3)$ &  &  & 6 (4)  & 100 (100)  & 7 (7)  & 12 (16) & 7 (7) & 23 (27) & 38 (44) & 30 (34) & 44 (46) \\
$CMU_3$ &  &  &  &  5 (5)&  100 (100)& 100 (100)& 100 (100)& 100 (100)& 100 (100)& 100 (100)& 100 (100)\\
$W_3(3, 2I_3)$ &  &  &  &  &  5 (5) & 45 (58) & 6 (5) & 32 (41)& 38 (58)& 41 (48) & 51 (59) \\
$IW_3(5, 3I_3)$ &  &  &  &  &  &  5 (6)  & 46 (62) & 17 (16) & 20 (16) & 16 (20) & 18 (28)\\
$W_3(3, K_3)$ &  &  &  &  &  &    & 4 (6) & 32 (46) & 49 (49) & 36 (47) & 43 (54)\\
$CMT_3(3, K_3)$ &  &  &  &  &  &   &   & 5 (6) & 6 (7)& 6 (6)& 4 (7) \\
$CMT_3(5, K_3)$ &  &  &  &  &  &   &   &  & 4 (4)& 6 (4) & 3 (3) \\
$CMT_3(3, I_3)$ &  &  &  &  &  &   &   &  &  & 6 (7)& 6 (7)\\
$CMT_3(5, I_3)$ &  &  &  &  &  &   &   &  &  &  & 6 (7)\\
\hline


\hline
&\multicolumn{11}{c}{$n=40, k=10$}\\
Distributions & $W_3(3, I_3)$ & $IW_3(3, I_3)$ & $CMT_3(1, I_3)$ &$CMU_3$& $W_3(3, 2I_3)$ & $IW_3(5, 3I_3)$ & $W_3(3, K_3)$ & $CMT_3(3, K_3)$ & $CMT_3(5, K_3)$  & $CMT_3(3, I_3)$ & $CMT_3(5, I_3)$\\ \hline
$W_3(3, I_3)$ & 5 (7) & 7 (6) & 6 (4) & 100 (100) & 11 (10) & 15 (15) & 10 (10) & 6 (13) & 13 (15) & 14 (12) & 18 (25) \\
$IW_3(3, I_3)$ & 17 (20) & 4 (4)  & 8 (10) & 100 (100)& 21 (25)  & 7 (4)& 19 (23)& 5 (8)& 5 (9)& 9 (9)& 6 (12)\\
$CMT_3(1, I_3)$ & 9 (7)& 4 (6)& 5 (6) & 100  (100)& 8  (8)& 9 (10)& 11 (11)& 8 (7)& 9 (14)& 9 (14)& 12 (16) \\
$CMU_3$ & 100 (100)& 100 (100)& 100 (100)&  5 (6) &  100 (100)& 100 (100)& 100 (100)& 100 (100)& 100 (100)& 100 (100)& 100 (100)\\
$W_3(3, 2I_3)$ & 3 (4)& 8 (8)& 5 (8)& 100 (100)&  5 (5)& 14 (17)& 8 (6)& 9 (16)& 22 (22)& 14 (18)& 17 (30)\\
$IW_3(5, 3I_3)$ &  27 (25)& 5 (4)& 8 (12)& 100 (100)& 42 (46) &  5 (6) & 42 (49)& 6 (6)& 10 (11)& 7 (12)& 13 (16)\\
$W_3(3, K_3)$ & 5 (5)& 7 (10)& 3 (4)& 100 (100)& 6 (4)& 20 (28)  & 6 (5)& 8 (13)& 17 (24)& 13 (17)& 18 (28)\\
$CMT_3(3, K_3)$ & 27 (31)& 9 (10)& 22 (24)& 100 (100)& 41 (40)& 14 (14)& 37 (38)& 4 (5)& 4 (5)& 6 (5)& 6 (6) \\
$CMT_3(5, K_3)$ & 32 (45)& 16 (16)& 27 (34)& 100 (100)& 41 (45)& 20 (20) & 36 (43) & 9 (7) & 3 (4) & 7 (7) & 6 (6) \\
$CMT_3(3, I_3)$ & 31 (33)& 16 (13)& 30 (30)& 100 (100)& 35 (34)& 20 (19) & 34 (40) & 7 (4) & 5 (8) & 5 (5)& 5 (5)\\
$CMT_3(5, I_3)$ & 43 (46)& 22 (20) & 32 (36) & 100 (100)& 39 (46) & 24 (23) & 40 (42) & 6 (7)& 5 (5)& 5 (4)& 7 (7)\\
\hline




\hline
&\multicolumn{11}{c}{$n=100, k=50$}\\
Distributions & $W_3(3, I_3)$ & $IW_3(3, I_3)$ & $CMT_3(1, I_3)$ &$CMU_3$& $W_3(3, 2I_3)$ & $IW_3(5, 3I_3)$ & $W_3(3, K_3)$ & $CMT_3(3, K_3)$ & $CMT_3(5, K_3)$  & $CMT_3(3, I_3)$ & $CMT_3(5, I_3)$\\ \hline
$W_3(3, I_3)$ & 5 (3) & 34 (51) & 6 (5) & 100 (100) & 13 (26)& 73 (84)& 21 (28)& 77 (79)& 91 (92)& 76 (86)& 95 (98)\\
$IW_3(3, I_3)$ &  & 5 (3) & 28 (51) & 100 (100)& 59 (75) & 8 (6)& 55 (74) & 26 (30) & 47 (51) & 35 (36) & 54 (63) \\
$CMT_3(1, I_3)$ &  &  & 6 (3) & 100  (100)& 11 (11)  & 48 (67) & 12 (17) & 75 (75) & 90 (91) & 75 (80) & 88 (92) \\
$CMU_3$ &  &  &  &  5 (5)&  100 (100)& 100 (100)& 100 (100)& 100 (100)& 100 (100)& 100 (100)& 100 (100)\\
$W_3(3, 2I_3)$ &  &  &  &  &  4 (4) & 98 (99)& 5 (6) & 89 (95) & 97 (98) & 93 (96) & 98 (98) \\
$IW_3(5, 3I_3)$ &  &  &  &  &  &  6 (6)  & 99 (100) & 30 (32) & 52 (53) & 39 (43) & 58 (63) \\
$W_3(3, K_3)$ &  &  &  &  &  &    & 6 (5) & 89 (92) & 96 (99) & 88 (95) & 98 (99) \\
$CMT_3(3, K_3)$ &  &  &  &  &  &   &   & 4 (5) & 8 (10) & 7 (6) & 11 (12) \\
$CMT_3(5, K_3)$ &  &  &  &  &  &   &   &  & 6 (6)& 6 (7)& 8 (6)\\
$CMT_3(3, I_3)$ &  &  &  &  &  &   &   &  &  & 5 (5) & 6 (9)\\
$CMT_3(5, I_3)$ &  &  &  &  &  &   &   &  &  &  & 7 (7)\\

\hline
&\multicolumn{11}{c}{$n=100, k=25$}\\
Distributions & $W_3(3, I_3)$ & $IW_3(3, I_3)$ & $CMT_3(1, I_3)$ &$CMU_3$& $W_3(3, 2I_3)$ & $IW_3(5, 3I_3)$ & $W_3(3, K_3)$ & $CMT_3(3, K_3)$ & $CMT_3(5, K_3)$  & $CMT_3(3, I_3)$ & $CMT_3(5, I_3)$\\ \hline
$W_3(3, I_3)$ & 5 (3) & 36 (40) & 7 (5) & 100 (100) & 6 (8) & 65 (66) & 3 (8) & 62 (70) & 86 (86) & 74 (72) & 89 (86) \\
$IW_3(3, I_3)$ & 14 (16) & 4 (5) & 17 (11) & 100 (100) & 24 (28) & 5 (8) & 28 (33) & 22 (21) & 38 (45) & 23 (24) & 43 (43) \\
$CMT_3(1, I_3)$ & 4 (6) & 23 (25) & 4 (4) & 100 (100) & 4 (6) & 28 (36) & 5 (6) & 61 (62) & 77 (76) & 71 (65) & 80 (82) \\
$CMU_3$ & 100 (100) & 100 (100) & 100 (100) & 100 (100) & 6 (4) &  100 (100) & 100 (100) & 100 (100) & 100 (100) & 100 (100) & 100 (100) \\
$W_3(3, 2I_3)$ & 18 (14) & 52 (57) & 13 (11) & 100 (100) & 6 (5) & 89 (90) & 4 (6) & 75 (80) & 92 (94) & 87 (85) & 94 (94) \\
$IW_3(5, 3I_3)$ & 46 (50) & 7 (4) & 32 (35) & 100 (100) & 76 (79) & 5 (6) & 73 (84) & 24 (24) & 52 (48) & 35 (35) & 49 (55) \\
$W_3(3, K_3)$ & 25 (31) & 57 (64) & 14 (14) & 100 (100) & 9 (5) & 89 (94) & 5 (6) & 78 (81) & 87 (91) & 79 (82) & 92 (94) \\
$CMT_3(3, K_3)$ & 46 (45) & 10 (12) & 36 (30) & 100 (100) & 45 (61) & 12 (17) & 59 (64) & 7 (5) & 6 (4) & 4 (4) & 10 (6)  \\
$CMT_3(5, K_3)$ & 54 (56) & 16 (22) & 66 (58) & 100 (100) & 74 (81) & 23 (19) & 74 (78) & 4 (7) & 4 (5) & 4 (3) & 7 (7) \\
$CMT_3(3, I_3)$ & 32 (52) & 13 (18) & 39 (42) & 100 (100) & 59 (57) & 14 (15) & 60 (62) & 5 (4) & 5 (5) & 3 (4) & 7 (5) \\
$CMT_3(5, I_3)$ & 60 (69) & 25 (28) & 64 (69) & 100 (100) & 77 (85) & 23 (29)  & 85 (75) & 7 (10) & 4 (5) & 5 (4) & 3 (4) \\
\hline
\end{tabular}
\end{adjustbox}
\end{table}

\section{Real data examples}\label{sec::realdata}
In this section, we present recent real-world data examples that clearly demonstrate the applicability of the proposed methodology.

\subsection{Example 1 - most recent cryptocurrency data}

In this example, most recent Bitcoin (BTC) and Ethereum (ETH) data are analyzed. The hourly BTC-USD and ETH-USD data for the period 01-01-2023 to 01-07-2023 is obtained from Gemini (\href{http://www.gemini.com}{http://www.gemini.com}). The close prices $X_t$ are considered. The hourly logarithmic returns $\log\frac{X_t}{X{t-1}}$ are computed for each hour. The sample covariance matrices are computed for each day, giving a total of $n=181$ observations. Multiple change points are detected using a binary segmentation algorithm (Algorithm \ref{BSAl}). We have used $Window=10$ and $NB=500$. We highlight the importance of properly choosing the parameter $Window$, as its large values cause the algorithm not to detect change points, while its small values cause the algorithm to detect unimportant points, especially in the analysis of volatile assets. Therefore, fine-tuning of the parameter $Window$ is necessary to provide optimal performance. We have analyzed the performance for $\gamma=1$ and $\gamma=0.5$, but no difference in change point estimated locations was observed. The important crypto-related news is given in Table \ref{matrixNews}. Note that no change points were detected in the period January - February, which is expected since crypto markets were mostly on the rise during this period of time. Bitcoin did suffer a mid-February drop but recovered quickly. Ethereum was unaffected \cite{januar, februar}. The information on change point location can prove useful in validating novel trading systems on the newest data. It can especially prove to be useful in automatically evaluating the conservativeness of a certain market strategy in conditions in which other means may prove to be futile due to the extreme underlying volatility of the assets in question.

\begin{table}[htbp]
\centering
\caption{Important cryptocurrency-related events on the day, day before or day after the detected change point}
\label{matrixNews}
\begin{tabular}{@{}lll@{}}
\hline
Change point date & Important news & Reference \\ \hline
2 March & Significant BTC and ETH drops on 3 March & \cite{mat1}\\
17 March & Bitcoin on the rise & \cite{mat2} \\
4 April & Significant drops on crypto markets & \cite{mat3} \\
18 April & Crypto market significantly falls & \cite{mat4} \\
2 May & Crypto prices on the rise & \cite{mat5} \\
17 May & Crypto trading should be treated as a form of gambling in the UK & \cite{mat6} \\
1 June & Cryptocurrencies reacting to debt ceiling deal & \cite{mat7} \\
16 June & Fall after Federal Reserve June meeting & \cite{mat8} \\ \hline
\end{tabular}
\end{table}

\subsection{Example 2 - performance on well-documented cases}
In this study, we analyze the data example provided in \cite{nas1} from a change point perspective. Specifically, we examine the behavior of the covariance structure of hourly logarithmic returns for a period of 15 days before and 15 days after well-documented drops in Bitcoin price that align with significant historical events \cite{fruehwirt2021cumulation}. Each period consists of $n=720$ data points, and the results are presented in Table \ref{pvals1hr}. However, the majority of the tests failed to detect the presence of change points both before and after the prominent events, with the exception of the month of February where increased volatility may have contributed to improved detection. Based on these findings, we conclude that this method is practically unusable for this particular analysis.

We conducted the change point analysis of 1-minute BTC (Bitcoin) \cite{BTCdata} and ETH (Ethereum) \cite{ETHdata} data. We selected two-day periods and calculated the logarithmic returns for each minute. Covariance matrices were computed for each hour, resulting in a total of $n=48$ covariance matrices, with 24 matrices for the first day and 24 for the second day. The results, presented in Table \ref{pvals1min}
indicate that our test successfully identified change points on the day before and on the day of the event's occurrence. The test never fails to detect the change point or lags with the detection. This finding holds potential practical significance in the implementation of stop losses in trading strategies \cite{kaminski2014stop}. Both tests, for $\gamma=0.5$ and $\gamma=1$, produced similar results.
\begin{table}[htbp]
\caption{p-values (change point location) in covariance structure before and after the Bitcoin important events - 1 minute data, $\gamma=0.5$ [$\gamma=1$].}
\label{pvals1min}
\centering
\begin{adjustbox}{width=\linewidth,center}
\begin{tabular}{@{}lp{2.3in}lllll@{}}
\hline
Date of event ($T_0$) & Event description & $p_{[T_0-2D, T_0-1D]}$ & $p_{[T_0-1D, T_0]}$ & $p_{[T_0,  T_0+1D]}$ &  $p_{[T_0-2D, T_0+1D]}$\\ \hline
8 November 2017 & Developers cancel splitting of Bitcoin. & 0.0980(27) [0.0940 (27)] & 0.0028 (41) [0.0016 (41)] & 0.2120 (17) [0.1080 (17)] &  0.0420 (65) [0 (65)]  \\
28 December 2017 & South Korea announces strong measures to regulate the trading of cryptocurrencies. & 0.0860 (39) [0.1400 (39)] & 0.0040 (26) [0 (26)]& 0.0080 (11) [0.0060 (25)] & 0.0800 (50) [0.020 (50)] \\
13 January 2018 & Announcement that 80\% of Bitcoin has been mined. & 0.0020 (7) [0 (7)] & 0 (7) [0 (7)] & 0.2060 (28) [0.1120 (28)] & 0 (31) [0 (31)] \\
30 January 2018 & Facebook bans advertisements promoting cryptocurrencies. & 0.1280 (33) [0.0800 (33)] & 0.0100 (37) [0 (37)] & 0.2220 (13) [0.1420 (17)] &  0 (61) [0 (61)] \\
7 March 2018 & The US Securities and Exchange Commission says it is necessary for crypto trading platforms to register. & 0 (41) [0.0020 (34)] & 0 (40) [0 (40)] & 0.1060 (16) [0.0040 (16)] & 0 (64) [0 (64)] \\
14 March 2018 & Google bans advertisements promoting cryptocurrencies. & 0.0060 (14) [0.0080 (14)] & 0.2260 (40) [0.3100 (33)] & 0.0040 (16)[0.0140 (16)] & 0.0180 (64)  [0.030 (64)]\\ \hline
\end{tabular}
\end{adjustbox}
\end{table}

\begin{table}[htbp]
\caption{p-values (change point location) in covariance structure before and after the Bitcoin important events - 1 hour data, $\gamma=0.5$ [$\gamma = 1$].}
    \label{pvals1hr}
    \begin{adjustbox}{width=\linewidth,center}
    \begin{tabular}{@{}llllp{2.9in}lll@{}}
\hline
Period I start date  & Period II start date& Date of event ($T_0$) & Period II end date & Event description & $p_{[T_0-30D, T_0]}$ & $p_{[T_0-15D, T_0+15D]}$ \\ \hline
9 October 2017 & 24 October 2017 & 8 November 2017 & 23 November 2017 & Developers cancel splitting of Bitcoin. & 0.29 (17) [0.20 (17)] & 0.32 (15) [0.20 (15)] \\
28 November 2017 & 13 December 2017 & 28 December 2017 & 12 January 2018 & South Korea announces strong measures to regulate trading of cryptocurrencies. & 0.30 (24) [0.26 (24)] & 0.59 (12) [0.58 (12)] \\
14 December 2017 & 28 December 2017 & 13 January 2018 & 28 January 2018 & Announcement that 80\% of Bitcoin has been mined. & 0.63 (11) [0.60 (11)] & 0.25 (19) [0.20 (19)]\\
31 December 2017 & 15 January 2018 & 30 January 2018 & 14 February 2018 & Facebook bans advertisements promoting, cryptocurrencies. & 0.59 (16) [0.41 (16)] & 0.14 (3) [0.31 (4)] \\
5 February 2018 & 20 February 2018 & 7 March 2018 & 22 March 2018 & The US Securities and Exchange Commission says it is necessary for crypto trading platforms to register. & 0 (3) [0 (6)] & 0 (15) [0 (15)]\\
12 February 2018 & 28 February 2018 & 14 March 2018 & 29 March 2018 & Google bans advertisements promoting cryptocurrencies. & 0 (23) [0 (23)] & 0.04 (7) [0.04 (7)]\\ \hline
\end{tabular}
    \end{adjustbox}
\end{table}


\subsection{Example 3 - high-dimensional data}
Note that this example is for the illustrative purposes only and it opens questions for further research.
The novel test is applied on the subset of German DAX index data. The following symbols are considered: 

ADS.DE, ALV.DE, BAS.DE, BAYN.DE, BEI.DE, BMW.DE, CON.DE, DB1.DE, DBK.DE, DHER.DE, DPW.DE, DTE.DE, DWNI.DE, EOAN.DE, FME.DE, FRE.DE, HEI.DE, HEN3.DE, IFX.DE, LHA.DE, LIN.DE, MRK.DE, MUV2.DE, RWE.DE, SAP.DE, SIE.DE, VNA.DE, and VOW3.DE. 

The period of 1-1-2022 to 1-1-2023 is considered, giving the total of $257$ trading days. The data are split in the sections of five consecutive trading days and the covariance matrix is estimated on every section, giving the total of $N=52$ observations. Since the dimensionality of data is high and the computational effort seems considerable, data reduction has to be performed prior to the test application. The Principal Component Analysis (PCA) is applied on the whole dataset. 
The explained variance (PCA EV) is given in Table \ref{DAXTB}. The reduction is performed by replacing a covariance matrix of the given dimension with a diagonal matrix, where the diagonal is populated with the largest eigenvalues. The largest eigenvalues (LE) per position are given in Table \ref{DAXTB} as well. From Table \ref{DAXTB}, it is clear that dimensions 2, 3, 4, and 5 can be considered. A binary search procedure is implemented on the reduced data, with $Window=5$ and $NB=500$. Every procedure gives the change point location $k=50$, suggesting that the single change point exists in the data and that the period from 11 to 17 December 2022 exhibits a potential change of behavior in the German market.

\begin{table}[htbp]
\centering
\caption{PCA explained variance and largest eigenvalues per dimension}
\label{DAXTB}
\begin{adjustbox}{width=\linewidth,center}
\begin{tabular}{llllllllllll}
\hline
Dim & PCA EV & LE & Dim & PCA EV & LE &  Dim & PCA EV & LE & Dim & PCA EV &  LE  \\\hline
1 & 0.6035 & 619.249 & 8 & 0.9665 & 1.57E-15 & 15 & 0.9926 & 2.15E-16 & 22 & 0.9985 & 5.70E-19   \\
2 & 0.7751 & 53.69442 & 9 & 0.9734 & 1.15E-15 & 16 & 0.9938 & 1.01E-16 & 23 & 0.9989 & 0  \\
3 & 0.8571 & 16.55289 & 10 & 0.9794 & 9.71E-16 & 17 & 0.9950 & 7.40E-17 & 24 & 0.9993 & 0  \\
4 & 0.8985 & 9.336243 & 11 & 0.9837 & 8.16E-16 & 18 & 0.9959 & 2.45E-17 & 25 & 0.9996 & 0  \\
5 & 0.9218 & 3.87E-14 & 12 & 0.9865 & 5.56E-16 & 19 & 0.9967 & 1.20E-17 & 26 & 0.9998 & 0  \\
6 & 0.9440 & 1.20E-14 & 13 & 0.9889 & 4.99E-16 & 20 & 0.9974 & 6.64E-18 & 27 & 0.9999 & 0  \\
7 & 0.9567 & 5.90E-15 & 14 & 0.9910 & 2.45E-16 & 21 & 0.9980 & 2.93E-18 & 28 & 1.0000 & 0 \\\hline 
\end{tabular}
\end{adjustbox}
\end{table}

\section*{Conclusion}

This study introduces a new class of tests for detecting change points in the distribution of orthogonally invariant matrix-valued random variables. The tests are based on the Hankel transform and have been shown to be effective in detecting change points, especially in cryptocurrency markets. This paper presents the theoretical properties of the tests and demonstrates their performance through simulation studies. Real data examples from cryptocurrency markets are also provided to illustrate the practical use of the tests.

There is a significant area for further research, particularly in exploring different types of integral transforms such as Laplace transforms or Fourier transforms. Additionally, the focus could be shifted towards handling high-dimensional data, ultra-high frequency data, and financial assets with lower volatility.

\section*{Declaration of interest}
The authors have no interest to declare.
\section*{Acknowledgements}
The work of B. Milo\v sevi\'c is  supported by the Ministry of Science, Technological Development and Innovations of the Republic of Serbia (the contract 451-03-47/2023-01/ 200104). In addition, the results of this paper are based upon  her work
from COST Action HiTEc-Text, functional and other high-dimensional data in econometrics: New models, methods, applications, CA21163, supported by COST (European Cooperation in Science and Technology).

\bibliographystyle{abbrv}  
\bibliography{references}

\begin{thebibliography}{10}

\bibitem{achcar2010non}
J.~A. Achcar, E.~R. Rodrigues, C.~D. Paulino, and P.~Soares.
\newblock {Non-homogeneous Poisson models with a change-point: an application
  to ozone peaks in Mexico city}.
\newblock {\em Environmental and Ecological Statistics}, 17:521--541, 2010.

\bibitem{andreou2002detecting}
E.~Andreou and E.~Ghysels.
\newblock {Detecting multiple breaks in financial market volatility dynamics}.
\newblock {\em Journal of Applied Econometrics}, 17(5):579--600, 2002.

\bibitem{andreou2009structural}
E.~Andreou and E.~Ghysels.
\newblock {Structural breaks in financial time series}.
\newblock {\em Handseries of financial time series}, pages 839--870, 2009.

\bibitem{baringhaus2015two}
L.~Baringhaus and D.~Kolbe.
\newblock {Two-sample tests based on empirical Hankel transforms}.
\newblock {\em Statistical Papers}, 56:597--617, 2015.

\bibitem{baringhaus2010empirical}
L.~Baringhaus and F.~Taherizadeh.
\newblock {Empirical Hankel transforms and its applications to goodness-of-fit
  tests}.
\newblock {\em Journal of Multivariate Analysis}, 101(6):1445--1457, 2010.

\bibitem{borovskikh1994theory}
Y.~V. Borovskikh and V.~S. Korolyuk.
\newblock {\em Theory of U-statistics}.
\newblock Kluwer, Dordrecht, 1994.

\bibitem{brodsky1993nonparametric}
E.~Brodsky and B.~S. Darkhovsky.
\newblock {\em Nonparametric methods in change point problems}, volume 243.
\newblock Springer Science \& Business Media, Berlin, 1993.

\bibitem{chakar2017robust}
S.~Chakar, E.~Lebarbier, C.~Lévy-Leduc, and S.~Robin.
\newblock {A robust approach for estimating change-points in the mean of an
  $\operatorname{AR}(1)$ process}.
\newblock {\em Bernoulli}, 23(2):1408 – 1447, 2017.

\bibitem{chen2012parametric}
J.~Chen and A.~K. Gupta.
\newblock {\em Parametric statistical change point analysis: with applications
  to genetics, medicine, and finance}.
\newblock Springer, Berlin, 2012.

\bibitem{contal1999application}
C.~Contal and J.~O'Quigley.
\newblock {An application of changepoint methods in studying the effect of age
  on survival in breast cancer}.
\newblock {\em Computational statistics \& data analysis}, 30(3):253--270,
  1999.

\bibitem{cuparic2022new}
M.~Cupari{\'c}, B.~Milo{\v{s}}evi{\'c}, and M.~Obradovi{\'c}.
\newblock {New consistent exponentiality tests based on V-empirical Laplace
  transforms with comparison of efficiencies}.
\newblock {\em Revista de la Real Academia de Ciencias Exactas, F{\'\i}sicas y
  Naturales. Serie A. Matem{\'a}ticas}, 116(1):42, 2022.

\bibitem{davis2006structural}
R.~A. Davis, T.~C.~M. Lee, and G.~A. Rodriguez-Yam.
\newblock {Structural break estimation for nonstationary time series models}.
\newblock {\em Journal of the American Statistical Association},
  101(473):223--239, 2006.

\bibitem{dehling2013non}
H.~Dehling, A.~Rooch, and M.~S. Taqqu.
\newblock {Non-parametric change-point tests for long-range dependent data}.
\newblock {\em Scandinavian Journal of Statistics}, 40(1):153--173, 2013.

\bibitem{fruehwirt2021cumulation}
W.~Fruehwirt, L.~Hochfilzer, L.~Weydemann, and S.~Roberts.
\newblock {Cumulation, crash, coherency: A cryptocurrency bubble wavelet
  analysis}.
\newblock {\em Finance Research Letters}, 40:101668, 2021.

\bibitem{fryzlewicz2014wild}
P.~Fryzlewicz.
\newblock {Wild binary segmentation for multiple change-point detection}.
\newblock {\em The Annals of Statistics}, 42(6):2243 – 2281, 2014.

\bibitem{mat2}
Y.~Gaur.
\newblock {Bitcoin rises 9.2\% to \$27,359}.
\newblock {\em Reuters}, 2023.

\bibitem{mat1}
K.~Ghosh.
\newblock {Bitcoin falls 5.2\% to \$22,253}.
\newblock {\em Reuters}, 2023.

\bibitem{giacomini2013warp}
R.~Giacomini, D.~N. Politis, and H.~White.
\newblock {A warp-speed method for conducting Monte Carlo experiments involving
  bootstrap estimators}.
\newblock {\em Econometric theory}, 29(3):567--589, 2013.

\bibitem{gupta2018matrix}
A.~K. Gupta and D.~K. Nagar.
\newblock {\em Matrix variate distributions}.
\newblock Chapman and Hall/CRC, London, 2018.

\bibitem{hadjicosta2020Gamma}
E.~Hadjicosta and D.~Richards.
\newblock {Integral transform methods in goodness-of-fit testing, I: the gamma
  distributions}.
\newblock {\em Metrika}, 83(7):733--777, 2020.

\bibitem{hadjicosta2020integral}
E.~Hadjicosta and D.~Richards.
\newblock {Integral transform methods in goodness-of-fit testing, II: the
  Wishart distributions}.
\newblock {\em Annals of the Institute of Statistical Mathematics},
  72:1317--1370, 2020.

\bibitem{hawkins2010nonparametric}
D.~M. Hawkins and Q.~Deng.
\newblock {A nonparametric change-point control chart}.
\newblock {\em Journal of Quality Technology}, 42(2):165--173, 2010.

\bibitem{henze2012goodness}
N.~Henze, S.~G. Meintanis, and B.~Ebner.
\newblock {Goodness-of-fit tests for the gamma distribution based on the
  empirical Laplace transform}.
\newblock {\em Communications in Statistics-Theory and Methods},
  41(9):1543--1556, 2012.

\bibitem{herz1955bessel}
C.~S. Herz.
\newblock {Bessel functions of matrix argument}.
\newblock {\em Annals of Mathematics}, pages 474--523, 1955.

\bibitem{hlavka2020change}
Z.~Hl{\'a}vka, M.~Hu{\v{s}}kov{\'a}, and S.~G. Meintanis.
\newblock {Change-point methods for multivariate time-series: paired vectorial
  observations}.
\newblock {\em Statistical Papers}, 61:1351--1383, 2020.

\bibitem{huvskova2006change}
M.~Hu{\v{s}}kov{\'a} and S.~G. Meintanis.
\newblock {Change Point Analysis based on Empirical Characteristic Functions:
  Empirical Characteristic Functions}.
\newblock {\em Metrika}, 63:145--168, 2006.

\bibitem{huvskova2006rank}
M.~Hu{\v{s}}kov{\'a} and S.~G. Meintanis.
\newblock {Change-point analysis based on empirical characteristic functions of
  ranks}.
\newblock {\em Sequential Analysis}, 25(4):421--436, 2006.

\bibitem{MATLAB}
T.~M. Inc.
\newblock {MATLAB version: 9.14.0.2206163 (R2023a)}, 2023.

\bibitem{kaminski2014stop}
K.~M. Kaminski and A.~W. Lo.
\newblock {When do stop-loss rules stop losses?}
\newblock {\em Journal of Financial Markets}, 18:234--254, 2014.

\bibitem{koev2006efficient}
P.~Koev and A.~Edelman.
\newblock {The efficient evaluation of the hypergeometric function of a matrix
  argument}.
\newblock {\em Mathematics of Computation}, 75(254):833--846, 2006.

\bibitem{ETHdata}
P.~Kottarathil.
\newblock {Ethereum Historical Dataset}.
\newblock
  \url{https://www.kaggle.com/datasets/prasoonkottarathil/ethereum-historical-dataset?select=ETH_1min.csv},
  2020.
\newblock Version 2. Accessed: 2023-04-09.

\bibitem{mat7}
I.~Lee and V.~Hajrić.
\newblock {Bitcoin Faces Fresh Challenges After Debt Deal Moves Forward,
  Citigroup Warns}.
\newblock {\em Bloomberg}, 2023.

\bibitem{li2012multiple}
S.~Li and R.~Lund.
\newblock {Multiple changepoint detection via genetic algorithms}.
\newblock {\em Journal of Climate}, 25(2):674--686, 2012.

\bibitem{nas1}
{\v{Z}}.~Luki{\'c} and B.~Milo{\v{s}}evi{\'c}.
\newblock {A novel two-sample test within the space of symmetric positive
  definite matrix distributions and its application in finance}, 2023.

\bibitem{mat4}
T.~Macheel.
\newblock {Bitcoin and ether fall as investors weigh persistent inflation and
  rising interest rates}.
\newblock {\em CNBC}, 2023.

\bibitem{mat8}
T.~Macheel.
\newblock {Bitcoin briefly drops below \$25,000, Tether’s stablecoin falls
  under its dollar peg}.
\newblock {\em CNBC}, 2023.

\bibitem{februar}
T.~Macheel.
\newblock {Bitcoin, ether on track for a positive February despite mid-month
  drop and fading 2023 risk rally}.
\newblock {\em CNBC}, 2023.

\bibitem{mat6}
K.~Makortoff.
\newblock {Cryptocurrency trading in UK should be regulated as form of
  gambling, say MPs}.
\newblock {\em The Guardian}, 2023.

\bibitem{meintanis2008testing}
S.~G. Meintanis and M.~Hu{\v{s}}kov{\'a}.
\newblock {Testing procedures based on the empirical characteristic functions
  II: k-sample problem, change point problem}.
\newblock {\em Tatra Mt. Math. Publ}, 39:235--243, 2008.

\bibitem{milovsevic2016new}
B.~Milo{\v{s}}evi{\'c} and M.~Obradovi{\'c}.
\newblock {New class of exponentiality tests based on U-empirical Laplace
  transform}.
\newblock {\em Statistical Papers}, 57(4):977--990, 2016.

\bibitem{muirhead2009aspects}
R.~J. Muirhead.
\newblock {\em Aspects of multivariate statistical theory}.
\newblock John Wiley \& Sons, 1982.

\bibitem{mat3}
A.~Pant.
\newblock {Crypto Price Today: Bitcoin below 28k, Ethereum and other tokens
  edge lower}.
\newblock {\em CNBCTV18}, 2023.

\bibitem{januar}
A.~J. Permal.
\newblock {Bitcoin pumped 43\% in January 2023! What to expect in February —
  Watch The Market Report live}.
\newblock {\em Cointelegraph}, 2023.

\bibitem{serfling2009approximation}
R.~J. Serfling.
\newblock {\em Approximation theorems of mathematical statistics}.
\newblock John Wiley \& Sons, New York, 2009.

\bibitem{suquet2009reproducing}
C.~Suquet.
\newblock {Reproducing kernel Hilbert spaces and random measures}.
\newblock pages 143--152. World Scientific, 2009.

\bibitem{tan2016nonparametric}
C.~Tan, X.~Shi, X.~Sun, and Y.~Wu.
\newblock {On nonparametric change point estimator based on empirical
  characteristic functions}.
\newblock {\em Science China Mathematics}, 59:2463--2484, 2016.

\bibitem{thies2018bayesian}
S.~Thies and P.~Moln{\'a}r.
\newblock {Bayesian change point analysis of Bitcoin returns}.
\newblock {\em Finance Research Letters}, 27:223--227, 2018.

\bibitem{tsukuda2014l2}
K.~Tsukuda and Y.~Nishiyama.
\newblock {On L2 space approach to change point problems}.
\newblock {\em Journal of Statistical Planning and Inference}, 149:46--59,
  2014.

\bibitem{mat5}
J.~Yang.
\newblock {Bitcoin Climbs Above \$28.5K as Investors Weigh Fresh Bank Woes,
  Cool Jobs Data}.
\newblock {\em Coindesk}, 2023.

\bibitem{BTCdata}
M.~Zielinski.
\newblock {Bitcoin Historical Data}.
\newblock
  \url{https://www.kaggle.com/datasets/mczielinski/bitcoin-historical-data},
  2021.
\newblock Version 7. Accessed: 2023-04-09.

\end{thebibliography}

\newpage
\section*{Appendix}
\begin{algorithm}
  \caption{Binary segmentation algorithm $BS(x, NB, \gamma, \alpha, Window)$}\label{BSAl}
  \begin{algorithmic}[1]
  \State Compute $n$=Length(x).
  \If{$n$<Window}
    \State Return NULL.
    \Else
    \State Compute $\mathcal{J}^{*}=\mathcal{J}_\gamma(x)$.
    \State Compute $x_{i_1}, x_{i_2}, \dots, x_{i_{NB}}$ - different permutations without replacement of $x$. 
    \State Estimate the distribution under the null hypothesis with 
    
    $\text{Null}=\mathcal{J}_\gamma(x_{i_1}), \mathcal{J}_\gamma(x_{i_2}), \dots, \mathcal{J}_\gamma(x_{i_{NB}})$.
    \State Estimate the p-value as $\text{pval}=\text{Mean}(\text{Null}>\mathcal{J}^*)$.
    \If{\text{pval}  $>\alpha$}
    \State Return NULL.
    \Else
    \State Compute the position for which the maximum in $\mathcal{J}^*$ is attained. Denote this position with $k^*$.
    \State Split $x$ into two subsamples $x_{h1}=x_1, \dots, x_{[n/2]}$ and 
    
    $x_{h2}=x_{[n/2]+1},x_{[n/2]+2}, \dots, x_n$.
    \State Return list 
    
    $[k^*, BS(x_{h1}, NB, \gamma, \alpha, Window),[n/2]+BS(x_{h2}, NB, \gamma, \alpha, Window)].$
    \EndIf
    \EndIf
  \end{algorithmic}
\end{algorithm}

\begin{algorithm}
  \caption{Warp-speed bootstrap permutation algorithm}\label{WSBA}
  \begin{algorithmic}[1]
    \State Sample $\textbf{x}=(x_1, \dots, x_{n_1})$ from $ F_X$ and $\mathbf{y}=(y_1, \dots, y_{n_2})$ from $F_Y$ and let $n=n_1+n_2$ , and create a pooled sample $\textbf{z}=(z_1, z_2, \dots, z_{n_1+n_2})=(x_1, \dots, x_{n_1}, y_1, y_2, \dots, y_{n_2})$.;
    \State Compute $\mathcal{I}_{n,  \gamma, \nu}:=\mathcal{I}_{n,  \gamma, \nu}(\textbf{z})$;
    \State Generate random permutations $\pi: \{1, 2, \dots, n_1+n_2\}\to \{1, 2, \dots, n_1+n_2\}$ and the corresponding bootstrap permutation sample $\textbf{z*}=(z_{\pi(1)}, \dots, z_{\pi(n_1)}, z_{\pi(n_1+1)}, \dots, z_{\pi(n_1+n_2)})$.
    \State Compute $\mathcal{I}^*_{n,  \gamma, \nu}:=\mathcal{I}_{n,  \gamma, \nu}(\textbf{z}*)$;
    \State Repeat steps 1-4 N times and obtain two sequences of statistics $\{\mathcal{I}_{n,  \gamma ,\nu}^{(j)}\}$ and $\{\mathcal{I}_{n,  \gamma,\nu}^{*(j)}\}$,  $j=1,...,N$; 
    \State Reject the null hypothesis for the $j$--sample ($j=1,...,N$), if $\mathcal{I}_{n,  \gamma, \nu}^{(j)}>c_\alpha$,  where $c_\alpha$ denotes the $(1-\alpha)\%$ quantile of the empirical distribution of the bootstrap test statistics  $(\mathcal{I}_{n,  \gamma, \nu}^{*(j)}, \ j=1,...,N)$.
     
  \end{algorithmic}
\end{algorithm}

\end{document}